\documentclass[11pt, onecolumn]{article}
\usepackage[top=1in, bottom=1in, left=1.25in, right=1.25in]{geometry}
\usepackage{amsfonts}
\usepackage{amsmath}
\usepackage{graphicx}
\usepackage{color, soul}
\usepackage{algorithm,algorithmic}
\usepackage{bm}
\usepackage{booktabs}
\usepackage{amssymb}

\usepackage[amsmath,thmmarks]{ntheorem}
\usepackage{theorem}

\newtheorem{pro}{Proposition}
\newtheorem{cor}{Corollary}
\newtheorem{lem}{Lemma}
\newtheorem{defi}{Definition}
\newtheorem{rem}{Remark}
\newtheorem{thm}{Theorem}

\theoremheaderfont{\sc}\theorembodyfont{\upshape}
\theoremstyle{nonumberplain}
\theoremseparator{}
\theoremsymbol{\rule{1ex}{1ex}}
\newtheorem{proof}{Proof}

\renewcommand{\arraystretch}{1.5}

\begin{document}
\title{Local Measurement and Reconstruction for Noisy Graph Signals}
\author{Xiaohan~Wang, Jiaxuan~Chen, and Yuantao~Gu%
\thanks{The authors are with Department of Electronic Engineering, Tsinghua University, Beijing 100084, P. R. China. The corresponding author of this work is Yuantao Gu (e-mail: gyt@tsinghua.edu.cn).}}

\date{Submitted April 5, 2015}

\maketitle

\begin{abstract}
The emerging field of signal processing on graph plays a more and more important role in processing signals and information related to networks. Existing works have shown that under certain conditions a smooth graph signal can be uniquely reconstructed from its decimation, i.e., data associated with a subset of vertices. However, in some potential applications (e.g., sensor networks with clustering structure), the obtained data may be a combination of signals associated with several vertices, rather than the decimation. In this paper, we propose a new concept of local measurement, which is a generalization of decimation. Using the local measurements, a local-set-based method named iterative local measurement reconstruction (ILMR) is proposed to reconstruct bandlimited graph signals. It is proved that ILMR can reconstruct the original signal perfectly under certain conditions. The performance of ILMR against noise is theoretically analyzed. The optimal choice of local weights and a greedy algorithm of local set partition are given in the sense of minimizing the expected reconstruction error. Compared with decimation, the proposed local measurement sampling and reconstruction scheme is more robust in noise existing scenarios.

\end{abstract}

\section{Introduction}
%\subsection{Signal Processing on Graph}
In recent years, graph-based signal processing has become an active research field due to the increasing demands for signal and information processing in irregular domains \cite{shuman_emerging_2013,sandryhaila_big_2014}.
For an $N$-vertex undirected graph $\mathcal{G}(\mathcal{V}, \mathcal{E})$, where $\mathcal{V}$ denotes the vertex set and $\mathcal{E}$ denotes the edge set, if a real number is associated with each vertex of $\mathcal{G}$, these numbers of all the vertices constitute a graph signal ${\bf f}\in\mathbb{R}^N$.
Potential applications of graph signal processing have been found in areas including sensor networks \cite{zhu_graph_2012}, semi-supervised learning \cite{gadde_active_2014}, image processing \cite{yang_gesture_2014}, and structure monitoring \cite{chen_bridge_2013}.

A lot of concepts and techniques for classical signal processing are extended to graph signal processing. Related problems on graph include graph signal filtering \cite{sandryhaila_discrete_2013, chen_adaptive_2013}, graph wavelets \cite{Coifman_Diffusion_2006, hammond_wavelets_2011,narang_perfect_2012}, graph signal compression \cite{zhu_approximating_2012, nguyen_downsampling_2015}, uncertainty principle \cite{agaskar_aspectral_2013}, graph signal coarsening \cite{liu_coarsening_2014, liu_graphcoarsening_2014}, multiresolution transforms \cite{ekambaram_multiresolution_2013,shuman_aframework_2013}, parametric dictionary learning \cite{thanou_parametric_2013}, graph topology learning \cite{dong_learning_2014}, graph signal sampling and reconstruction \cite{anis_towards_2014, narang_localized_2013, wang_iterative_2014, wang_local_2014}, and distributed algorithms \cite{wang_distributed_2014, chen_distributed_2015}.

\subsection{Motivation and Related Works}
It is a natural problem to reconstruct smooth signals from partial observations on a graph in practical applications \cite{sandryhaila_discrete_2013, chen_adaptive_2013, narang_signal_2013}. For data gathering in sensor networks, sometimes only part of the nodes transmit data due to the limited bandwidth or energy. According to the smoothness of data, the missing entries can be estimated from the received ones, which can be modeled as the reconstruction of smooth signals on graph from decimation. Especially, for a sensor network with clustering structure, the collected data within a cluster are aggregated by the cluster head, which plays the role as a local measurement and can be naturally obtained. Using the measured data from all the clusters to retrieve the raw data of all the nodes can be modeled as a problem of smooth graph signal reconstruction from local measurements, which is a linear combination of the signal amplitudes in a cluster of vertices. This problem is studied in this work for the first time.

There have been several works focusing on the theory of exactly reconstructing a bandlimited graph signal from its decimation.
Sufficient conditions for unique reconstruction of bandlimited graph signals from decimation are given for normalized \cite{pesenson_sampling_2008} and unnormalized Laplacian \cite{fuhr_poincare_2013}.
In \cite{anis_towards_2014}, a necessary and sufficient condition on the cutoff frequency is established and the bandwidth is estimated based on the concept of spectral moments.
Several algorithms are proposed to reconstruct graph signals from decimation.
In \cite{narang_localized_2013}, an algorithm named iterative least square reconstruction (ILSR) is proposed and the tradeoff between data-fitting and smoothness is also considered.
Two more efficient algorithms named iterative weighting reconstruction (IWR) and iterative propagating reconstruction (IPR) are proposed in \cite{wang_local_2014} with much faster convergence.

As far as we know, there is no work on reconstructing graph signals from local measurements.
The idea of local measurements can be traced back to
time-domain nonuniform sampling \cite{marvasti_nonuniform_2001}, or irregular sampling \cite{feichtinger_theory_1994, grochenig_adiscrete_1993}, which has a close relationship with graph signal sampling and reconstruction.
For the signals in time-domain \cite{grochenig_reconstruction_1992, feichtinger_theory_1994}, shift-invariant space \cite{aldroubi_nonuniform_2002}, or on manifolds \cite{pesenson_poincare_2004,feichtinger_recovery_2004}, based on the theoretical results of signal reconstruction from samples, there has been extended works on reconstructing signals from local averages. However, there is no such work on graph-signal-related problems.

\subsection{Contributions}
In this paper, we first generalize the sampling scheme for graph signals from \emph{decimation} to \emph{local measurement}.
Based on this scheme, we then propose a new method named iterative local measurement reconstruction (ILMR) to reconstruct the original signal from limited measurements. It is proved that the bandlimited signals can always be exactly reconstructed from its local measurements if certain conditions are satisfied. Moreover, we demonstrate that the traditional decimation scheme, which samples by vertex, and its corresponding reconstruction method is a special case of this work. Based on the performance analysis of ILMR, we find that the local measurement is more robust than decimation in noise scenario. As a consequence, the optimal local weights in different noisy environment are discussed. The proposed sampling scheme has several advantages. First, it will benefit in the situation where local measurements are easier to obtain than the samples of specific vertices. Second, the proposed local measurement and reconstruction is more robust against noise.

%The main contributions of this work include
%\begin{itemize}
%\item
%A generalized sampling scheme local measurement;
%\item
%An algorithm ILMR to reconstruct graph signals from local measurements;
%\item
%Analysis of ILMR against noise, and the optimal local weights and local sets.
%\end{itemize}

This paper is organized as follows. In section II, the basis of graph signal processing and some existing algorithms for reconstructing graph signals from decimation are reviewed. The generalized sampling scheme, i.e. local measurement, is proposed in section III. In section IV, the reconstruction algorithm ILMR is proposed and its convergence is proved. In section V, the reconstruction performance in noise scenario is studied, and the optimal choice of local weight and local set partition is discussed. Experimental results are demonstrated in section VI, and the paper is concluded in section VII.

\section{Preliminaries}
\subsection{Laplacian-based Graph Signal Processing and Bandlimited Graph Signals}

The Laplacian \cite{chung_spectral_1997} of an $N$-vertex undirected graph $\mathcal{G}$ is defined as
$$
{\bf L=D-A},
$$
where ${\bf A}$ is the adjacency matrix of $\mathcal{G}$, and ${\bf D}$ is the degree matrix, which is a diagonal matrix whose entries are the degrees of the corresponding vertices.

Since $\mathcal{G}$ is undirected, its Laplacian is a symmetric and positive semi-definite matrix, and all of the eigenvalues of ${\bf L}$ are real and nonnegative. If $\mathcal{G}$ is connected, there is only one zero eigenvalue. Denote the eigenvalues of ${\bf L}$ as
$0=\lambda_1<\lambda_2\le\cdots\le\lambda_N$, and the corresponding eigenvectors as $\{{\bf u}_k\}_{1\le k\le N}$. The eigenvectors can also be regarded as graph signals on $\mathcal{G}$.

The Laplacian ${\bf L}: \mathbb{R}^N\rightarrow\mathbb{R}^N$ is an operator on the space of graph signals on $\mathcal{G}$,
$$
({\bf Lf})(u)=\sum_{v\in\mathcal{V}, u\sim v}\!\!\left(f(u)-f(v)\right),\quad \forall u\in\mathcal{V},
$$
where $f(u)$ denotes the entry of ${\bf f}$ associated with vertex $u$, and $u\sim v$ denotes that there is an edge between vertices $u$ and $v$. The Laplacian can be viewed as a kind of differential operator between vertices and their neighbors. Therefore, among the eigenvectors of ${\bf L}$, those associated with small eigenvalues have similar amplitudes on connected vertices, while the eigenvectors associated with large eigenvalues vary fast on the graph. In other words, eigenvectors associated with small eigenvalues are smooth or denote low-frequency components of signals on $\mathcal{G}$.

For graph Fourier transform \cite{hammond_wavelets_2011}, the eigenvectors $\{{\bf u}_k\}_{1\le k\le N}$ are regarded as the Fourier basis of the frequency-domain, and the eigenvalues $\{\lambda_k\}_{1\le k\le N}$ are regarded as frequencies.
The graph Fourier transform is
$$
\hat{f}(k)=\langle {\bf f}, {\bf u}_k\rangle=\sum_{i=1}^Nf(i)u_k(i),
$$
where $\hat{f}(k)$ is the strength of frequency $\lambda_k$.

Similar to its counterpart in time-domain, if a graph signal ${\bf f}$ is smooth on $\mathcal{G}$, ${\bf f}$ may be uniquely determined by its entries on a limited number of sampled vertices. Based on the graph Laplacian, the smoothness of a graph signal is usually described as being within a bandlimited subspace. A graph signal ${\mathbf f}\in\mathbb{R}^N$ is $\omega$-bandlimited if
$${\mathbf f}\in PW_{\omega}(\mathcal{G})\triangleq\text{span}\{{\bf u}_i|\lambda_i\le\omega\},$$
which is called Paley-Wiener space on $\mathcal{G}$ \cite{pesenson_sampling_2008}.

\subsection{Reconstruction from Decimation of Bandlimited Graph Signals}

There have been theoretical analysis and algorithms on the reconstruction from decimation of bandlimited graph signals.
Existing results show that ${\mathbf f}\in PW_{\omega}(\mathcal{G})$ can be uniquely reconstructed from its entries $\{f(u)\}_{u\in \mathcal{S}}$ on a sampling vertex set $\mathcal{S}\subseteq \mathcal{V}$ under certain conditions.
%Typical reconstruction algorithms include ILSR and IPR.
An important concept of \emph{uniqueness set} is introduced in \cite{pesenson_sampling_2008}.

\begin{defi}[uniqueness set {\cite{pesenson_sampling_2008}}]
 A set of vertices $\mathcal{S}\subseteq \mathcal{V}(\mathcal{G})$ is a uniqueness set for space $PW_{\omega}(\mathcal{G})$ if it holds for all ${\mathbf f}, {\mathbf g} \in PW_{\omega}(\mathcal{G})$ that $f(u)$ equals $g(u)$ for all $u\in\mathcal{S}$ implies ${\mathbf f}$ equals ${\mathbf g}$.
\end{defi}

Then iterative least square reconstruction (ILSR), a reconstruction algorithm from decimation of graph signal is proposed, which can be written in the following equivalent form.
\begin{thm}[ILSR {\cite{narang_localized_2013}}]\label{thm:ILSR}
If the sampling set $\mathcal{S}$ is a uniqueness set for $PW_{\omega}(\mathcal{G})$, then the original signal ${\bf f}$ can be reconstructed using the decimation $\{f(u)\}_{u\in\mathcal{S}}$
by the following ILSR method,
\begin{align}
{\bf f}^{(0)}&=\mathcal{P}_{\omega}\left(\sum_{u\in \mathcal{S}}f(u)\bm{\delta}_{u}\right),\nonumber\\
{\mathbf f}^{(k+1)}&={\mathbf f}^{(k)}+\mathcal{P}_{\omega}\left(\sum_{u\in \mathcal{S}}(f(u)-f^{(k)}(u))\bm{\delta}_{u}\right)\nonumber,
\end{align}
where $\mathcal{P}_{\omega}(\cdot)$ is the projection operator onto $PW_{\omega}(\mathcal{G})$, and $\bm{\delta}_u$ is a Dirac delta function whose entries satisfying
\begin{equation}\label{defdelta}
\delta_u(v)=
\begin{cases}
1, & v=u; \\
0, & v\neq u.
\end{cases}
\end{equation}
\end{thm}

To accelerate the convergence, an algorithm named iterative propagating reconstruction (IPR) is proposed, which is based on an important concept of \emph{local sets}.

\begin{defi}[local sets {\cite{wang_local_2014}}]\label{deflocalset1}
For a sampling set $\mathcal{S}$ on a graph $\mathcal{G}(\mathcal{V},\mathcal{E})$, assume that $\mathcal{V}$ is divided into disjoint local sets $\{\mathcal{N}(u)\}_{u\in \mathcal{S}}$ associated with the sampled vertices. For each $u\in\mathcal{S}$, denote the subgraph of $\mathcal{G}$ restricted to $\mathcal{N}(u)$ by $\mathcal{G}_{\mathcal{N}(u)}$, which is composed of vertices in $\mathcal{N}(u)$ and edges between them in $\mathcal{E}$. For each $u\in\mathcal{S}$, its local set satisfies $u\in\mathcal{N}(u)$, and the subgraph $\mathcal{G}_{\mathcal{N}(u)}$ is connected.  Besides, $\{\mathcal{N}(u)\}_{u\in \mathcal{S}}$ should satisfy
$$
\bigcup_{u\in \mathcal{S}} \mathcal{N}(u)=\mathcal{V}
\text{   and   }
\mathcal{N}(u)\cap \mathcal{N}(v)=\emptyset, \quad \forall u, v\in\mathcal{S}, u\neq v.
$$
\end{defi}

The property of a local set is measured by \emph{maximal multiple number} and \emph{radius}, as follows.
\begin{defi}[maximal multiple number {\cite{wang_local_2014}}]\label{defimmn}
Denoting $\mathcal{T}(u)$ as the shortest-path tree of $\mathcal{G}_{\mathcal{N}(u)}$ rooted at $u$,
for $v\sim u$ in $\mathcal{T}(u)$, $\mathcal{T}_u(v)$ is the subtree composed by $v$ and its descendants in $\mathcal{T}(u)$.
The \emph{maximal multiple number} of $\mathcal{N}(u)$ is
$$
K(u)=\max_{v\sim u \text{ in } \mathcal{T}(u)}|\mathcal{T}_u(v)|.
$$
\end{defi}

\begin{defi}[radius {\cite{wang_local_2014}}]\label{defradius}
The \emph{radius} of $\mathcal{N}(u)$ is the maximal distance of vertex in $\mathcal{G}_{\mathcal{N}(u)}$ from $u$, denoted as
$$
R(u)=\max_{v\in \mathcal{N}(u)}\text{dist}(v,u).
$$
\end{defi}

\begin{thm}[IPR {\cite{wang_local_2014}}]\label{thm:IPR}
For a given sampling set $\mathcal S$ and associated local sets $\{\mathcal{N}(u)\}_{u\in\mathcal{S}}$ on a graph $\mathcal{G}(\mathcal{V},\mathcal{E})$, $\forall {\mathbf f}\in PW_{\omega}(\mathcal{G})$, if $\omega$ is less than $1/Q_{\rm max}^2$, ${\mathbf f}$ can be reconstructed by its decimation $\{f(u)\}_{u\in \mathcal{S}}$ through the IPR method
\begin{align}
{\mathbf f}^{(0)}&=\mathcal{P}_{\omega}\left(\sum_{u\in \mathcal{S}}f(u)\bm{\delta}_{\mathcal{N}(u)}\right),\nonumber\\
{\mathbf f}^{(k+1)}&={\mathbf f}^{(k)}+\mathcal{P}_{\omega}\left(\sum_{u\in \mathcal{S}}(f(u)-f^{(k)}(u))\bm{\delta}_{\mathcal{N}(u)}\right),\nonumber
\end{align}
%where $Q_{\rm max}$ is a parameter determined by $\mathcal{S}$ and $\{\mathcal{N}(u)\}_{u\in\mathcal{S}}$,
where $$Q_{\rm max}=\max_{u\in\mathcal{S}}\sqrt{K(u)R(u)},$$
and $\bm{\delta}_{\mathcal{N}(u)}$ denotes the graph signal with entries
$$
\delta_{\mathcal{N}(u)}(v)=
\begin{cases}
1, & v\in \mathcal{N}(u);\\
0, & v\notin \mathcal{N}(u).
\end{cases}
$$
\end{thm}
IPR converges faster than ILSR, because in each iteration, IPR updates a larger increment than ILSR by utilizing a propagation to local sets.

\section{Local Measurement: A Generalized Sampling Scheme}

We consider a new sampling scheme of measuring by local sets. In this scheme, all the vertices in a graph is partitioned into disjoint clusters. In each cluster, there is no specific sampling vertex, but all vertices in this cluster contribute to produce a measurement.     For this purpose, \emph{centerless local sets} are first introduced based on Definition \ref{deflocalset1}.

\begin{defi}[centerless local sets]\label{deflocalset2}
For a graph $\mathcal{G}(\mathcal{V},\mathcal{E})$, assume that $\mathcal{V}$ is divided into disjoint local sets $\{\mathcal{N}_i\}_{i\in\mathcal{I}}$, where $\mathcal I$ denotes the index set of divisions. Each subgraph $\mathcal{G}_{\mathcal{N}_i}$, which denotes the subgraph of $\mathcal{G}$ restricted to $\mathcal{N}_i$, is connected.  Besides, $\{\mathcal{N}_i\}_{i\in \mathcal{I}}$ should satisfy
$$
\bigcup_{i\in \mathcal{I}} \mathcal{N}_i=\mathcal{V}
\text{    and    }
\mathcal{N}_i\cap \mathcal{N}_j=\emptyset, \quad \forall i, j\in\mathcal{I}, ~ i\neq j.
$$
\end{defi}

One should notice that the centerless local sets play important roles in the proposed generalized sampling scheme, while the local sets do not in traditional decimation scheme. In the decimation scheme, the local sets are designed for specific reconstruction algorithms and have no effect in the sampling process. However, in the generalized sampling scheme, the centerless local sets are elaborated for sampling and determine the performance of reconstruction, which will be discussed in section \ref{secnoise}.

To evaluate the partition of a graph, the \emph{diameter} of a centerless local set is defined and will be utilized in next section.
\begin{defi}[diameter]\label{diameter}
For a centerless local set $\mathcal{N}_i$, its diameter is defined as the largest distance of two vertices in $\mathcal{G}_{\mathcal{N}_i}$, i.e.,
$$
D_i = \max_{u,v\in\mathcal{N}_i}{\text{dist}}(u,v).
$$
\end{defi}

In order to produce a measurement from specific centerless local set, a \emph{local weight} is defined to balance the contribution of all vertices in this set and to obstruct the energy from other part of the graph.

\begin{defi}[local weight]
A local weight $\bm{\varphi}_i\in \mathbb{R}^N$ associated with a centerless local set $\mathcal{N}_i$ satisfies
$$
\varphi_i(v)
\begin{cases}
\ge 0, v\in \mathcal{N}_i\\
= 0, v\notin \mathcal{N}_i
\end{cases}
$$
and
$$
\sum_{v\in \mathcal{N}_i}\varphi_i(v)=1.
$$
\end{defi}

We highlight that the weight is \emph{local} rather than \emph{global} comes from some natural observations. It is partially because that locality and local operations are basic features of graphs and complex networks. Moreover, signal processing on graph may be dependent on distributed implementation, where local operations are more feasible than global ones.

Finally, we arrive at the definition of \emph{local measurement} by linearly combining the signal amplitudes in each centerless local set using preassigned local weights.

\begin{defi}[local measurement]
For given centerless local sets and the associated local weights $\{(\mathcal{N}_i,\bm{\varphi}_i)\}_{i\in \mathcal{I}}$, a set of local measurements for a graph signal $\bf f$ is $\{f_{\bm{\varphi}_i}\}_{i\in \mathcal{I}}$, where
$$
f_{\bm{\varphi}_i}\triangleq\langle {\mathbf f}, \bm{\varphi}_i\rangle=\sum_{v\in \mathcal{N}_i}f(v)\varphi_i(v).
$$
\end{defi}

\begin{figure}[t]
\begin{center}
\includegraphics[width=11cm]{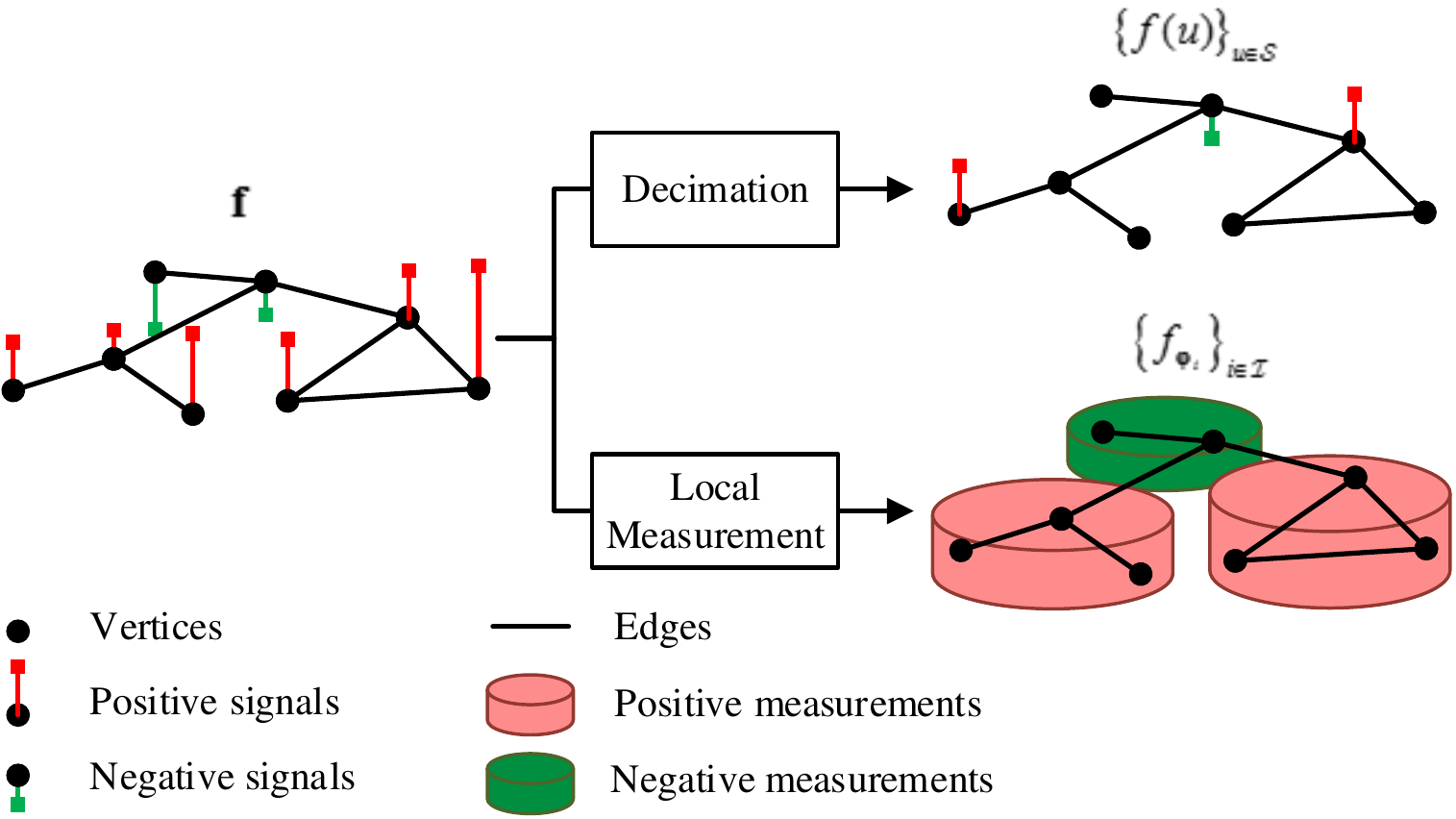}
\caption{An illustration of traditional sampling (decimation) scheme versus generalized sampling (local measurement) scheme. For each centerless local set, a local measurement is produced by a linear combination of signal amplitudes associated with vertices within this set.}
\label{decivsmeas}
\end{center}
\end{figure}

The sampling schemes of decimation and of local measurement are visualized in Fig. \ref{decivsmeas}. Compared with decimation in previous works \cite{wang_local_2014, pesenson_sampling_2008}, local measurement can be regarded as a generalized sampling scheme.
The local measurement $\{f_{\bm{\varphi}_i}\}_{i\in\mathcal{I}}$ is to obtain a linear combination of the signal in each local set, while the decimation $\{f(u)\}_{u\in\mathcal{S}}$ is to obtain the signal on selected vertices in the sampling set $\mathcal{S}$.
Both sampling schemes take the inner products of the original signal and specified local weights.
Decimation can be regarded as a special case of local measurement, with all the weights in each centerless local set assigned to only one vertex, i.e., the sampled vertex.

\section{ILMR: Reconstruct Signal from Local Measurements}
We will show that under certain conditions the original signal ${\mathbf f}$ can be uniquely and exactly reconstructed from the local measurements $\{f_{\bm{\varphi}_i}\}_{i\in\mathcal{I}}$.

First of all, an operator is defined based on centerless local sets and the associated local weights.

\begin{defi}\label{limitedpropagation}
For given centerless local sets and the associated weights $\{(\mathcal{N}_i, \bm{\varphi}_i)\}_{i\in\mathcal{I}}$ on a graph $\mathcal{G}(\mathcal{V},\mathcal{E})$, an operator ${\bf G}$ is defined by
\begin{align}
{\bf G}{\mathbf f}&=\mathcal{P}_{\omega}\left(\sum_{i\in \mathcal{I}}\langle {\mathbf f}, \bm{\varphi}_i\rangle\bm{\delta}_{\mathcal{N}_i}\right)\label{localmeasprop1}\\
&= \sum_{i\in \mathcal{I}}\langle {\mathbf f}, \bm{\varphi}_i\rangle\mathcal{P}_{\omega}(\bm{\delta}_{\mathcal{N}_i}),\label{localmeasprop2}
\end{align}
where $\bm{\delta}_{\mathcal{N}_i}$ is defined as
\begin{equation}
\delta_{\mathcal{N}_i}(v)=
\begin{cases}
1, & v\in \mathcal{N}_i;\\
0, & v\notin \mathcal{N}_i.
\end{cases}\label{deltaNi}
\end{equation}
\end{defi}

For a graph signal, the proposed operator is to calculate the local measurement in each centerless local set, then to assign the local measurement to all the vertices in that set, and finally to filter out the component beyond the bandwidth, i.e., \eqref{localmeasprop1}. Equivalently, it denotes a linear combination of all low-frequency part of $\{\bm{\delta}_{\mathcal{N}_i}\}_{i\in\mathcal{I}}$, with the combination coefficients as the local measurements of corresponding local sets, i.e., \eqref{localmeasprop2}.

The following lemma shows that the proposed operator is bounded in $PW_{\omega}(\mathcal{G})$ under certain conditions.

\begin{lem}\label{lemma1}
For given centerless local sets and the associated weights $\{(\mathcal{N}_i, \bm{\varphi}_i)\}_{i\in\mathcal{I}}$, $\forall{\mathbf f}\in PW_{\omega}(\mathcal{G})$, the following inequality holds,
$$
\|{\mathbf f}-{\bf G}{\mathbf f}\|\le C_{\rm max}\sqrt{\omega}\|{\mathbf f}\|,
$$
where
$$
C_{\rm max}=\max_{i\in\mathcal{I}}\sqrt{|\mathcal{N}_i|D_i},
$$
and $|\cdot|$ denotes cardinality.
\end{lem}

The proof of Lemma \ref{lemma1} is postponed to section \ref{proof1}. Lemma \ref{lemma1} shows that the operator $({\bf I-G})$ is a contraction mapping in $PW_{\omega}(\mathcal{G})$ if $\omega$ is less than $1/C_{\rm max}^2$.

Based on Lemma \ref{lemma1}, it is shown in Proposition \ref{pro1} that the original signal can be reconstructed from its local measurements.

\begin{pro}\label{pro1}
For given centerless local sets and the associated weights $\{(\mathcal{N}_i, \bm{\varphi}_i)\}_{i\in\mathcal{I}}$, $\forall {\mathbf f}\in PW_{\omega}(\mathcal{G})$, where $\omega$ is less than $1/C_{\rm max}^2$, ${\mathbf f}$ can be reconstructed from its local measurements $\{f_{\bm{\varphi}_i}\}_{i\in \mathcal{I}}$ through an iterative local measurement reconstruction (ILMR) algorithm in Table \ref{alg},
\begin{table}[t]
\renewcommand{\arraystretch}{1.2}
\caption{Iterative Local Measurement Reconstruction.}\label{alg}
\begin{center}
\begin{tabular}{l}
\toprule[1pt]
{\bf Input:} \hspace{0.5em} Graph $\mathcal{G}$, cutoff frequency $\omega$, centerless local sets $\{\mathcal{N}_i\}_{i\in\mathcal{I}}$, \\
\hspace{3.5em} local weights $\{\bm{\varphi}_i\}_{i\in\mathcal{I}}$, local measurements $\{f_{\bm{\varphi}_i}\}_{i\in\mathcal{I}}$;\\
{\bf Output:} \hspace{0.5em} Interpolated signal ${\bf f}^{(k)}$;\\
\hline
{\bf Initialization:}\\
\begin{minipage}{0.45\textwidth}
\begin{equation}\label{ILMR:init}
{\mathbf f}^{(0)}={\mathcal P}_{\omega}\left(\sum_{i\in \mathcal{I}}f_{\bm{\varphi}_i}\bm{\delta}_{\mathcal{N}_i}\right);
\end{equation}
\end{minipage}\\
{\bf Loop:}\\
\begin{minipage}{0.45\textwidth}
\begin{equation}\label{ILMR:loop}
{\mathbf f}^{(k+1)}={\mathbf f}^{(k)}+{\mathcal P}_{\omega}\left(\sum_{i\in \mathcal{I}}(f_{\bm{\varphi}_i}-\langle {\mathbf f}^{(k)}, \bm{\varphi}_i\rangle)\bm{\delta}_{\mathcal{N}_i}\right);
\end{equation}
\end{minipage}\\
{\bf Until:}\hspace{0.5em} The stop condition is satisfied.\\
\bottomrule[1pt]
\end{tabular}
\end{center}
\end{table}
with the error at the $k$th iteration satisfying
$$
\|{\mathbf f}^{(k)}-{\mathbf f}\|\le \gamma^{k}\|{\mathbf f}^{(0)}-{\bf f}\|,
$$
where
\begin{equation}\label{defgamma}
\gamma =C_{\rm max}\sqrt{\omega}.
\end{equation}
\end{pro}

\begin{proof}
According to the definition of ${\bf G}$, the iteration (\ref{ILMR:loop}) can be rewritten as
\begin{equation}\label{iter}
{\mathbf f}^{(k+1)}={\mathbf f}^{(k)}+{\bf G}({\mathbf f}-{\mathbf f}^{(k)}).
\end{equation}
Note that ${\mathbf f}\in PW_{\omega}(\mathcal{G})$ and ${\mathbf f}^{(k)}\in PW_{\omega}(\mathcal{G})$ for any $k$, then ${\mathbf f}^{(k)}-{\mathbf f}\in PW_{\omega}(\mathcal{G})$.
As a consequence of Lemma \ref{lemma1},
$$
\|{\mathbf f}^{(k+1)}-{\mathbf f}\|=\|({\mathbf f}^{(k)}-{\mathbf f})-{\bf G}({\mathbf f}^{(k)}-{\mathbf f})\|\le \gamma\|{\mathbf f}^{(k)}-{\mathbf f}\|.
$$
\end{proof}

Proposition \ref{pro1} shows that a signal ${\mathbf f}$ is uniquely determined and can be reconstructed by its local measurements $\{f_{\bm{\varphi}_i}\}_{i\in \mathcal{I}}$ if $\{\bm{\varphi}_i\}_{i\in \mathcal{I}}$ are known.
The quantity $(f_{\bm{\varphi}_i}-\langle {\mathbf f}^{(k)}, \bm{\varphi}_i\rangle)$ is the estimate error between the original measurement and the reconstructed measurement at the $k$th iteration.
According to (\ref{iter}), in each iteration of ILMR, the new increment of the interpolated signal is obtained by first assigning the estimate error to all vertices in the associated centerless local sets, and then projecting it onto the $\omega$-bandlimited subspace.

Except for the difference of decimation and local measurement, the basic idea of ILMR is similar to that of IPR \cite{wang_local_2014}, which is an algorithm of reconstructing graph signal from decimation. The procedures of IPR and ILMR in each iteration are illustrated in Fig. \ref{algorithm}.
One may find that in the assignment or propagating step, ILMR assigns the estimate errors of local measurements to vertices within the local sets, while IPR propagates the estimate errors of the decimated signal on the sampled vertices to other vertices in the local sets. In fact, ILMR degenerates to IPR if the local weight concentrates on only one vertex (the sampled vertex) in each local set, in which case the local measurement degenerates to decimation.

The sufficient conditions and error bounds for ILMR and IPR are also different. Suppose the (centerless) local sets divisions in ILMR and IPR are exactly the same, i.e. the sampling set $\mathcal{S}$ in IPR can be written as $\{u_i\}_{i\in\mathcal{I}}$, where $\mathcal{I}$ is the index set in ILMR, then $\mathcal{N}_i$ equals $\mathcal{N}(u_i)$ for all $i\in\mathcal{I}$. According to Definition \ref{defimmn} and \ref{defradius}, we have $R(u_i) \le D_i$ and $K(u_i) \le |\mathcal{N}(u_i)|=|\mathcal{N}_i|$. Therefore, $C_{\rm max}$ is not less than $Q_{\rm max}$. It implies that a more strict condition is needed to reconstruct a graph signal accurately from local measurements than to reconstruct it from decimation. However, since both sufficient conditions in Theorem \ref{thm:IPR} and Proposition \ref{pro1} are not tight and there is still room for refinement, such a comparison only provides a rough analysis.

%The sufficient conditions and error bounds for ILMR and IPR are also different. The sufficient condition of ILMR is more strict than that of IPR, which means a more strict condition is needed for ILMR to reconstruct a graph signal accurately from local measurements. Besides, for signals with the same bandwidth, the decay parameter $\gamma$ of ILMR is larger than that of IPR, which means IPR may converge faster than ILMR.
%However, since the estimate of the sufficient conditions and error bounds have some room for improvement, such a comparison only gives a rough analysis. In fact, experimental results show that for signals with the same bandwidth, ILMR converge faster than IPR.

\begin{figure*}[t]
\begin{center}
\includegraphics[width=16cm]{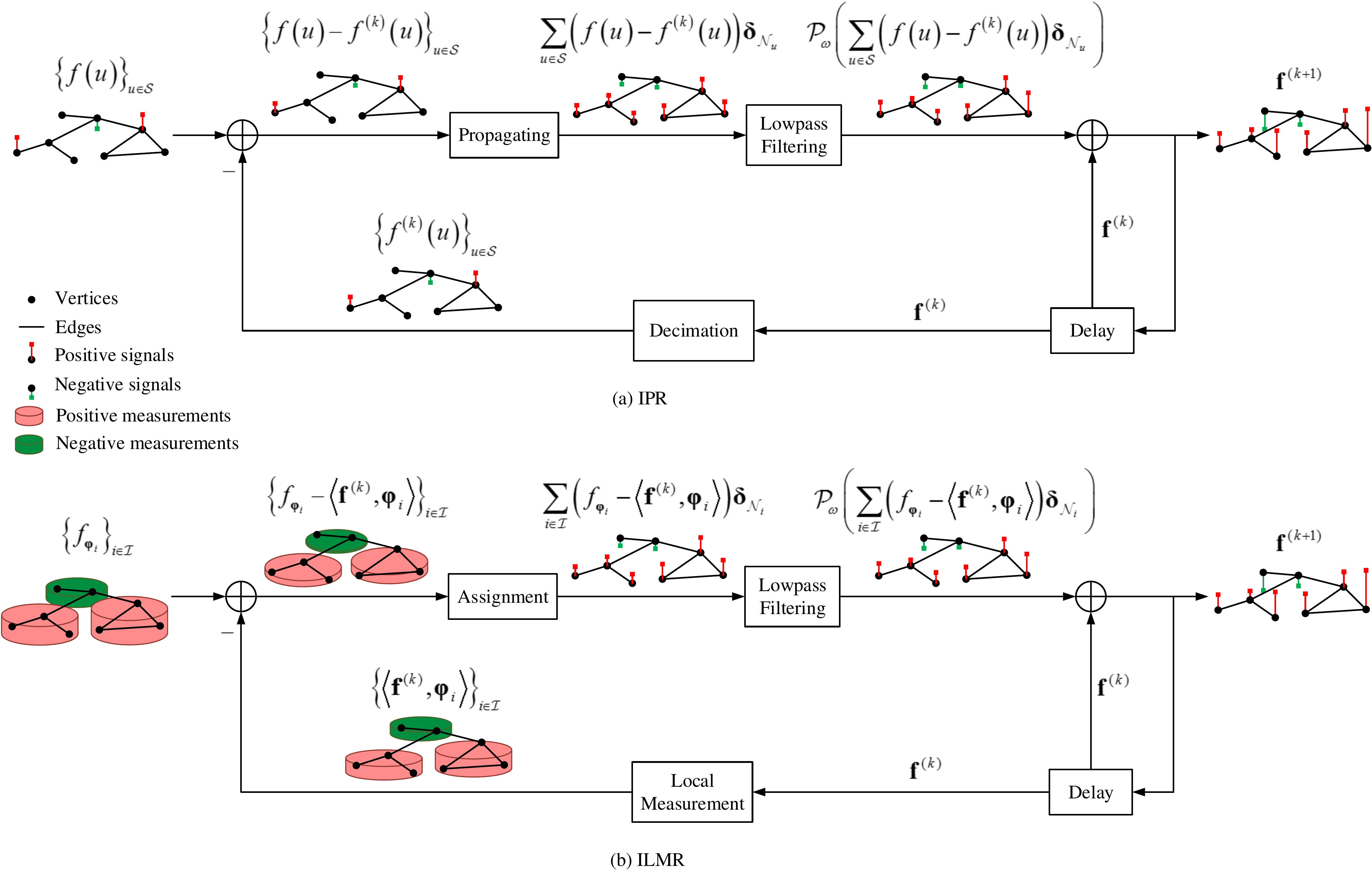}
\caption{The procedures of IPR and ILMR. The former algorithm is to reconstruct a bandlimited signal from decimation, while the latter reconstructs a signal from local measurements. Essentially, ILMR becomes IPR if the local weights concentrate on only one vertex of each local set, in which case local measurement degenerates to decimation. }
\label{algorithm}
\end{center}
\end{figure*}

\begin{rem}
For potential applications, if the local measurements come from the result of some repeatable physical operations, the local weights are even not necessarily known when conducting ILMR.
In detail, if $\{\bm{\varphi}_i\}_{i\in\mathcal{I}}$ is unknown but fixed, i.e., the local measurement operation in Fig. \ref{algorithm}(b) is a black box, $\langle {\mathbf f}^{(k)}, \bm{\varphi}_i\rangle$ may also be obtained by conducting the physical operations in each iteration. Therefore, the original signal can still be reconstructed by ILMR without exactly knowing $\{\bm{\varphi}_i\}_{i\in\mathcal{I}}$. This is a rather interesting result, and may facilitate graph signal reconstruction in specific scenarios.
\end{rem}

\section{Performance Analysis}\label{secnoise}

In this section, we study the error performance of ILMR when the original signal is corrupted by additive noise.  We  first derive the reconstruction error for incorrect measurement. Then the expected reconstruction error is calculated under the assumption of independent Gaussian noises and the optimal local weight is obtained in the sense of minimizing the expected reconstruction error bound. Finally, in a special case of \emph{i.i.d.} Gaussian perturbation, a greedy method for the centerless local sets partition and the selection of optimal local weights are provided.

\subsection{Reconstruction Error in Noise Scenario}

Suppose the observed signal associated with each vertex is corrupted by additive noise.
The corrupted signal is denoted as $\tilde{\bf f}={\bf f+n}$, where ${\bf n}$ denotes the noise.
In the $k$th iteration of ILMR, the corrupted local measurements $\{\langle \tilde{\mathbf f}, \bm{\varphi}_i\rangle\}_{i\in\mathcal{I}}$ is utilized to produce the temporary reconstruction of $\tilde{\bf f}^{(k)}$.

The following lemma gives a reconstruction error bound of $\tilde{\bf f}^{(k)}$.
\begin{pro}\label{pro2}
For given centerless local sets and the associated weights $\{(\mathcal{N}_i, \bm{\varphi}_i)\}_{i\in\mathcal{I}}$, ${\mathbf f}\in PW_{\omega}(\mathcal{G})$ is corrupted by additive noise ${\bf n}$.
If $\omega$ is less than $1/C_{\rm max}^2$, in the $k$th iteration the output of ILMR using the corrupted local measurements $\{\langle \tilde{\mathbf f}, \bm{\varphi}_i\rangle\}_{i\in\mathcal{I}}$ satisfies
\begin{align}\label{tfk-f}
\|\tilde{\bf f}^{(k)}-{\bf f}\|
\le\frac{\tilde{n}}{1-\gamma}+\gamma^{k+1}\left(\|{\mathbf f}\|+\|{\bf n}\|\right),
\end{align}
where $\gamma$ is defined as (\ref{defgamma}), $\tilde{n}$ is defined as
\begin{equation}\label{tilden}
\tilde{n}=\sum_{i\in \mathcal{I}}\sqrt{|\mathcal{N}_i|}\cdot|n_i|,
\end{equation}
and $n_i$ is the equivalent noise of centerless local set $\mathcal{N}_i$, defined as
\begin{equation}\label{ni}
n_i=\langle {\bf n}, \bm{\varphi}_i\rangle=\sum_{v\in\mathcal{N}_i}n(v)\varphi_i(v).
\end{equation}
\end{pro}

The proof of Proposition \ref{pro2} is postponed to section \ref{proof2}.

From (\ref{tfk-f}) it can be seen that in the noise scenario the reconstruction error in controlled by the sum of two parts. The former one is a weighted sum of the equivalent noise of all the local sets, while the latter one is decaying with the increase of iteration number.
The former part is crucial as the iteration goes on. Thus minimizing the former part, which is determined by both partition of centerless local sets and local weights, may improve the performance of ILMR in the noise scenario.

\subsection{Gaussian Noise and Optimal Local Weights}\label{subsecoptimalweight}
For a given partition $\{\mathcal{N}_i\}_{i\in\mathcal{I}}$, some prior knowledge of unknown noise ${\bf n}$ may bring the possibility to design optimal local weights. For simplicity the noises associated with different vertices are assumed to be independent.

Suppose the noise follows zero-mean Gaussian distribution, i.e., ${\bf n}\sim \mathcal{N}({\bf 0},{\bf \Sigma})$, where ${\bf \Sigma}$ is a diagonal matrix and the noise of vertex $v$ satisfies $n(v)\sim\mathcal{N}(0,\sigma^2(v))$.
Then $\tilde{n}$ defined in (\ref{tilden}) is a random variable.

For centerless local set $\mathcal{N}_i$, according to (\ref{ni}), the equivalent noise $n_i$ also follows a Gaussian distribution
$n_i\sim\mathcal{N}(0,\sigma_i^2)$,
where
\begin{equation}\label{sigmai}
\sigma_i^2=\sum_{v\in\mathcal{N}_i}\sigma^2(v)\varphi_i^2(v).
\end{equation}
Then $|n_i|$ follows the half-normal distribution with its expectation satisfying
$$
{\rm E}\left\{|n_i|\right\}=\sigma_i\sqrt{\frac{2}{\pi}}.
$$

According to (\ref{tilden}), the expectation of $\tilde{n}$ is
\begin{equation}\label{Etilden}
{\rm E}\{\tilde{n}\}=\sqrt{\frac{2}{\pi}}\sum_{i\in\mathcal{I}}\sqrt{|\mathcal{N}_i|}\sigma_i.
\end{equation}
Then the following corollary is ready to obtain.

\begin{cor}\label{cor1}
For given centerless local sets and the associated weights $\{(\mathcal{N}_i, \bm{\varphi}_i)\}_{i\in\mathcal{I}}$, the original signal ${\mathbf f}\in PW_{\omega}(\mathcal{G})$, assuming the noise associated with vertex $v$ follows independent Gaussian distribution $\mathcal{N}(0,\sigma^2(v))$,
if $\omega$ is less than $1/C_{\rm max}^2$, the expected reconstruction error of ILMR in the $k$th iteration satisfies
\begin{equation}\label{Etfk-f}
{\rm E}\left\{\|\tilde{\bf f}^{(k)}-{\bf f}\|\right\}
\le\frac{1}{1-\gamma}\sqrt{\frac{2}{\pi}}\sum_{i\in\mathcal{I}}\sqrt{|\mathcal{N}_i|}\sigma_i
+\mathcal{O}\left(\gamma^{k+1}\right),
\end{equation}
where $\gamma$ is defined as (\ref{defgamma}), and $\sigma_i$ is defined as (\ref{sigmai}).
\end{cor}

Corollary \ref{cor1} is ready to prove by plugging (\ref{sigmai}) and (\ref{Etilden}) in the expectation of (\ref{tfk-f}).

By minimizing the right hand side of (\ref{Etfk-f}), the optimal choice of local weights\footnote{In fact, the optimal local weights can also be studied in other criterions, e.g. the fastest convergence. However, in this work we only consider in the sense of minimizing the expected reconstruction error bound.} can be derived.

\begin{cor}\label{optweights}
For given division of centerless local sets $\{\mathcal{N}_i\}_{i\in\mathcal{I}}$, if the noises associated with the vertices are independent and follow zero-mean Gaussian distributions $n(v)\sim\mathcal{N}(0,\sigma^2(v))$, then the optimal local weights $\{\bm{\varphi}_i\}_{i\in\mathcal{I}}$ are
\begin{equation}\label{optimalweight}
\varphi_i(v)=
\begin{cases}
\displaystyle{\frac{(\sigma^2(v))^{-1}}{\sum_{v\in\mathcal{N}_i}(\sigma^2(v))^{-1}}}, & v\in \mathcal{N}_i;\\
0, & v\notin \mathcal{N}_i.
\end{cases}
\end{equation}
\end{cor}

\begin{proof}
Minimizing the right hand side of (\ref{Etfk-f}) is equivalent to minimizing $\sigma_i$ for each local set $\mathcal{N}_i$.
By the Cauchy-Schwarz inequality, one has
$$
\left(\sum_{v\in\mathcal{N}_i}(\sigma^2(v))^{-1}\!\right)\sigma_i^2
=\left(\sum_{v\in\mathcal{N}_i}(\sigma^2(v))^{-1}\!\right)\!\!\left(\sum_{v\in\mathcal{N}_i}\sigma^2(v)\varphi_i^2(v)\!\right)
\ge \left(\sum_{v\in\mathcal{N}_i}\varphi_i(v)\!\right)^2=1.
$$
Therefore,
\begin{equation}\label{sigmaimin}
\sigma_i^2\ge \frac{1}{\sum_{v\in\mathcal{N}_i}(\sigma^2(v))^{-1}}.
\end{equation}
The equality of (\ref{sigmaimin}) holds if and only if (\ref{optimalweight}) is satisfied.
\end{proof}

The above analysis shows that in the sense of minimizing the expected reconstruction error, the optimal local weight associated with vertex $v$ within $\mathcal{N}_i$ is inversely proportional to the noise variance of $v$. This is evident because more information are reserved in the sampling process if a larger local weight is assigned to a vertex with smaller noise variance. However, it should be noted that compared with the optimal local measurement, assigning all the weights in $\mathcal{N}_i$ to the vertex with the smallest noise variance, i.e. the optimal decimation, is not the best choice.
In fact, the optimal choice of local measurements is consistent with the well-known inverse variance weighting in statistics \cite{lipsey_practical_2001}.

Therefore, local measurement may reduce the disturbance of noise and reconstruct the original signal more precisely. In other words, for given partition of centerless local sets, graph signal reconstruction from local measurements with the optimal weights may always perform better than reconstruction from decimation, even when the vertices with the smallest noise variance are chosen in the latter sampling scheme.

\subsection{A Special Case of Independent and Identical Distributed Gaussian Noise}

Specifically, if the noise variances are the same for all the vertices, i.e., $\sigma(v)$ equals $\sigma$ for any $v\in\mathcal{V}$, $\tilde{n}$ can be approximately written in a more explicit form. For $\mathcal{N}_i$, the optimal local weight is equal for all the vertices in $\mathcal{N}_i$. Thus $\varphi_i(v)$ equals $1/|\mathcal{N}_i|$ for $v\in\mathcal{N}_i$, and in this case, $\sqrt{|\mathcal{N}_i|}n_i$ follows a Gaussian distribution,
$$
\sqrt{|\mathcal{N}_i|}n_i\sim\mathcal{N}(0,\sigma^2).
$$
Then $\sqrt{|\mathcal{N}_i|}\cdot|n_i|$ follows the half-normal distribution with the same parameter $\sigma$.
The above analysis shows that each term of the sum in (\ref{tilden}) follows independent and identical half-normal distribution, with its expectation and variance satisfying
\begin{align*}
{\rm E}\left\{\sqrt{|\mathcal{N}_i|}\cdot|n_i|\right\}&=\sigma\sqrt{\frac{2}{\pi}},\\
\text{Var}\left\{\sqrt{|\mathcal{N}_i|}\cdot|n_i|\right\}&=\sigma^2\left(1-\frac{2}{\pi}\right).
\end{align*}
Because the number of local sets $|\mathcal{I}|$ is always relatively large, by the central limit theorem,
$\tilde{n}$ follows a Gaussian distribution approximately,
$$
\tilde{n}\sim\mathcal{N}\left(|\mathcal{I}|\sigma\sqrt{\frac{2}{\pi}}, |\mathcal{I}|\sigma^2\left(1-\frac{2}{\pi}\right)\right).
$$

Then we have the following corollary.
\begin{cor}\label{cor2}
For given centerless local sets $\{\mathcal{N}_i\}_{i\in\mathcal{I}}$ and the associated weights $\varphi_i(v)=1/|\mathcal{N}_i|$ for $v\in\mathcal{N}_i$, the original signal ${\mathbf f}\in PW_{\omega}(\mathcal{G})$, assuming the noise associated with each vertex follows \emph{i.i.d} Gaussian distribution $\mathcal{N}(0,\sigma^2)$,
if $\omega$ is less than $1/C_{\rm max}^2$, the expected reconstruction error of ILMR in the $k$th iteration satisfies
\begin{align}\label{Etfk-f2}
{\rm E}\left\{\|\tilde{\bf f}^{(k)}-{\bf f}\|\right\}
\le\frac{|\mathcal{I}|\sigma}{1-\gamma}\sqrt{\frac{2}{\pi}}
+\mathcal{O}\left(\gamma^{k+1}\right),
\end{align}
where $\gamma$ is defined as (\ref{defgamma}).
\end{cor}

According to (\ref{Etfk-f2}), the error bound is affected by the number of centerless local sets $|\mathcal{I}|$. A division with fewer sets may reduce the expected reconstruction error.
However, it should be noted that the number of centerless local sets cannot be too small to satisfy the condition
$$
\gamma=C_{\rm max}\sqrt{\omega} = \max_{i\in\mathcal{I}}\sqrt{|\mathcal{N}_i|D_i\omega}<1,
$$
which is determined by the cutoff frequency of the original graph signal.
Besides, the factor $1/(1-\gamma)$ in (\ref{Etfk-f2}) implies that a smaller $C_{\rm max}$, which leads to a smaller $\gamma$, also reduces the error bound.
A roughly calculation can be given to balance the two factors.
If there are not too many vertices in each $\mathcal{N}_i$, we have that $C_{\rm max}$ approximates to $N_{\rm max}$, where $N_{\rm max}$ is the largest cardinality of centerless local sets. Since $N_{\rm max}|\mathcal{I}|$ approximates to $N$, we have
$$
\frac{1}{1-\gamma}|\mathcal{I}|\approx \frac{1}{1-\sqrt{\omega}N_{\rm max}}\cdot\frac{N}{N_{\rm max}}.
$$
To minimize the above quantity, a near optimal $N_{\rm max}$ is
\begin{equation}\label{Nmaxapp}
N_{\rm max} = \frac{1}{2\sqrt{\omega}},
\end{equation}
i.e., $\gamma$ approximates to $1/2$.
It provides a strategy to partition centerless local sets. For given cutoff frequency $\omega$, an approximated $N_{\rm max}$ can be chosen according to (\ref{Nmaxapp}), then the graph is divided into local sets to make sure that $|\mathcal{N}_i|$ is not more than $N_{\rm max}$ and the number of local sets is as small as possible.

For a given $N_{\rm max}$, a greedy algorithm is proposed to make the division of centerless local sets, as shown in Table \ref{algLocalSet}. The greedy algorithm is to iteratively remove connected vertices with the smallest degrees from the original graph into the new set, until the cardinality of the new set reaches $N_{\rm max}$ or there is no connected vertex. The reason for choosing the smallest-degree vertex is that such a vertex is more likely on the border of a graph.

\begin{table}[t]
\renewcommand{\arraystretch}{1.2}
\caption{A greedy method to partition centerless local sets with maximal cardinality.}\label{algLocalSet}
\begin{center}
\begin{tabular}{l}
\toprule[1pt]
{\bf Input:} \hspace{0.5em} Graph $\mathcal{G(V,E)}$, Maximal cardinality $N_{\text{max}}$;\\
{\bf Output:} \hspace{0.5em} Centerless local sets $\{\mathcal{N}_i\}_{i\in\mathcal{I}}$;\\
\hline
{\bf Initialization:}\hspace{0.5em} $i=0$;\\
{\bf Loop Until:} $\mathcal{V}=\emptyset$\\
\hspace{1.3em} 1) Find one vertex with the smallest degree in $\mathcal{G}$,\\
\hspace{3.9em} $\displaystyle u=\arg\min_{v\in \mathcal{V}}d_{\mathcal{G}}(v)$;\\
\hspace{1.3em} 2) $i=i+1$, $\mathcal{N}_i=\{u\}$;\\
\hspace{1.3em} 3) Obtain the neighbor set of $\mathcal{N}_i$, \\
\hspace{3.9em} $\mathcal{S}_i=\{v\in\mathcal{G}|v\sim w, w\in\mathcal{N}_i, v\notin\mathcal{N}_i\}$;\\
\hspace{1.3em} {\bf Loop Until:} $|\mathcal{N}_i|=N_{\text{max}}$ or $\mathcal{S}_i=\emptyset$\\
\hspace{2.6em} 4) Find one vertex with the smallest degree in $\mathcal{S}_i$,\\
\hspace{5.2em} $\displaystyle u=\arg\min_{v\in \mathcal{S}_i}d_{\mathcal{G}}(v)$;\\
\hspace{2.6em} 5) $\mathcal{N}_i=\mathcal{N}_i\cup\{u\}$;\\
\hspace{2.6em} 6) Update $\mathcal{S}_i=\{v\in\mathcal{G}|v\sim w, w\in\mathcal{N}_i, v\notin\mathcal{N}_i\}$;\\
\hspace{1.3em} {\bf End Loop}\\
\hspace{1.3em} 7) Remove the edges, $\displaystyle \mathcal{E}=\mathcal{E}\backslash\{(p,q)|p\in\mathcal{N}_i,q\in\mathcal{V}\}$;\\
\hspace{1.3em} 8) Remove the vertices, $\displaystyle \mathcal{V}=\mathcal{V}\backslash\mathcal{N}_i$ and $\mathcal{G=G(V,E)}$;\\
{\bf End Loop}\\
\bottomrule[1pt]
\end{tabular}
\end{center}
\end{table}

%\section{Relationship with Time Domain Results}

\section{Experiments}

We choose the Minnesota road graph \cite{gleich_matlabbgl}, which has $2640$ vertices and $6604$ edges, to verify the proposed generalized sampling scheme and reconstruction algorithm. The bandlimited signals for reconstruction are  generated by removing the high-frequency component of random signals, whose entries are drawn from \emph{i.i.d.} Gaussian distribution. The centerless local sets are generated by the greedy method in Table \ref{algLocalSet} using given $N_{\rm max}$. Five kinds of local weights are tested including
\begin{enumerate}
\item
uniform weight, where $\varphi_i(v)$ equals $1/|\mathcal{N}_i|, \forall v\in\mathcal{N}_i$;
\item
random weight, where
$$\varphi_i(v) = \frac{\varphi^\prime_i(v)}{\sum_{u\in \mathcal{N}_i}\varphi^\prime_i(u)}, \quad \forall v\in\mathcal{N}_i, \varphi^\prime_i(u)\sim\mathcal U(0,1);$$
\item
Dirac delta weight, where ${\bm \varphi}_i$ equals ${\bm \delta}_u$ for a randomly chosen $u\in\mathcal{N}_i$;
\item
the optimal weight, where
$$
\varphi_i(v) = \frac{(\sigma^2(v))^{-1}}{\sum_{v\in\mathcal{N}_i}(\sigma^2(v))^{-1}}, \quad \forall v\in \mathcal{N}_i;
$$
\item
the optimal Dirac delta weight, where ${\bm \varphi}_i$ equals ${\bm \delta}_u$ for
$$u=\arg\min_{u\in\mathcal{N}_i}\sigma^2(u).
$$
\end{enumerate}
Notice that case 3) and case 5) degenerate ILMR to IPR.

\subsection{Convergence of ILMR}

In the first experiment, the convergence of the proposed ILMR is verified for various centerless local sets partition and local weights. The graph is divided into $709$ and $358$ centerless local sets for $N_{\rm max}$ equals $4$ and $8$, respectively. Three kinds of local weights are tested including case 1), 2), and 3). The averaged convergence curves are plotted in Fig. \ref{exp1} for $100$ randomly generated original graph signals. According to Fig. \ref{exp1}, the convergence is accelerated when the graph is divided into more local sets and has a smaller $N_{\rm max}$. It is ready to understand because more local sets will bring more measurements and increase the sampling rate, which provides more information in the reconstruction. According to (\ref{defgamma}), for the same $\omega$, a smaller $N_{\rm max}$ leads to a smaller $\gamma$, and guarantees a faster convergence. The experimental result also shows that in the noise-free scenario, reconstruction with uniform weight converges slightly faster than that with random weight. However, both above cases converge much faster than reconstruction with Dirac delta weight. This means that local-measurement-based ILMR behaves better than decimation-based IPR by combining the signals on different vertices properly.

\begin{figure}[t]
\begin{center}
\includegraphics[width=9cm]{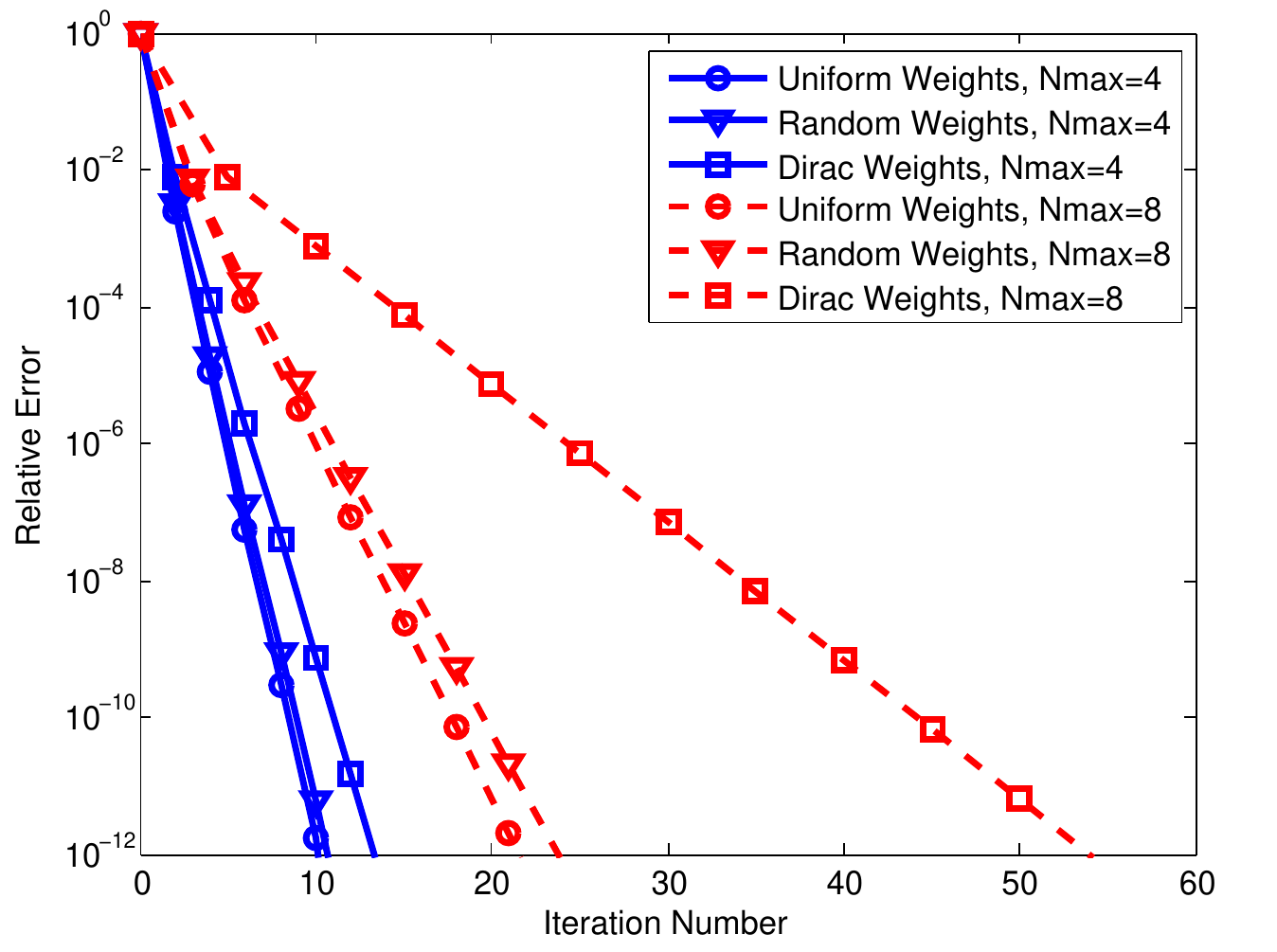}
\caption{The convergence behavior of ILMR for various division of centerless local sets and different local weights.}
\label{exp1}
\end{center}
\end{figure}

\subsection{Optimal Local Weights for Gaussian Noise}
\begin{figure}[t]
\begin{center}
\includegraphics[width=9cm]{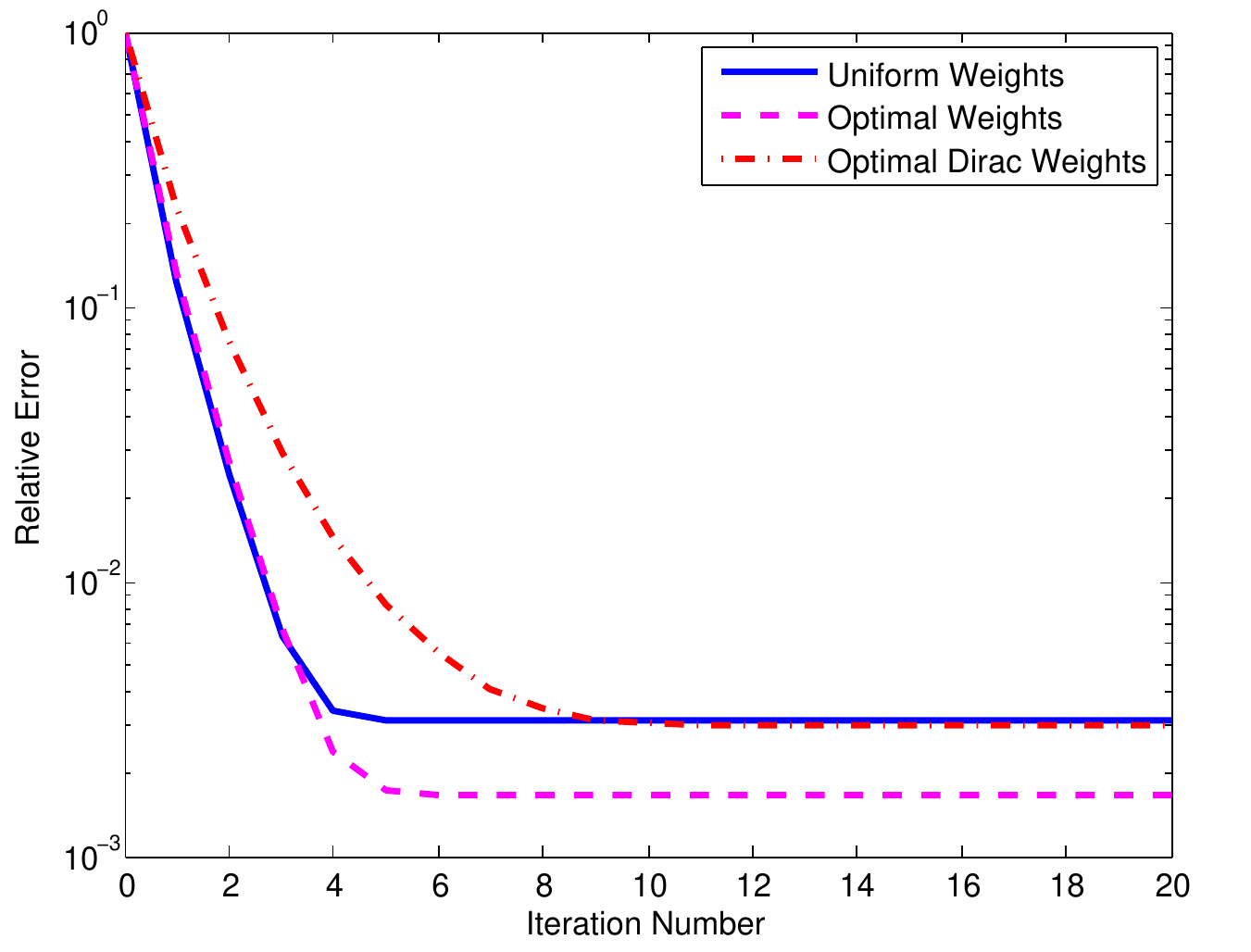}
\caption{The convergence curves of reconstruction with uniform weights, the optimal weights, and optimal Dirac delta weights when independent zero-mean Gaussian noise is added to each vertex. }
\label{exp3}
\end{center}
\end{figure}

In this experiment, independent zero-mean Gaussian noise is added to each vertex with different variance.
The original signal is normalized with unit norm. All of the vertices are randomly divided into three groups with the standard deviations of the noise chosen as $\sigma$ equals $1\times 10^{-4}$, $2\times 10^{-4}$, and $5\times 10^{-4}$, respectively. The graph is partitioned into $358$ centerless local sets with $N_{\rm max}$ equals $8$. Three kinds of local weights are tested including case 1), 4), and 5). The averaged convergence curves are illustrated in Fig. \ref{exp3} for $100$ randomly generated original graph signals. One may read that the steady-state relative error with the optimal weight is smaller than those with uniform weight and the optimal Dirac delta weight. The experimental result verifies the analysis in section \ref{subsecoptimalweight}. It implies that a better selection of local weights can reduce the reconstruction error if the noise variances on vertices are different.

\subsection{Performance against Independent and Identical Distributed Gaussian Noise}

\begin{figure}[t]
\begin{center}
\includegraphics[width=9cm]{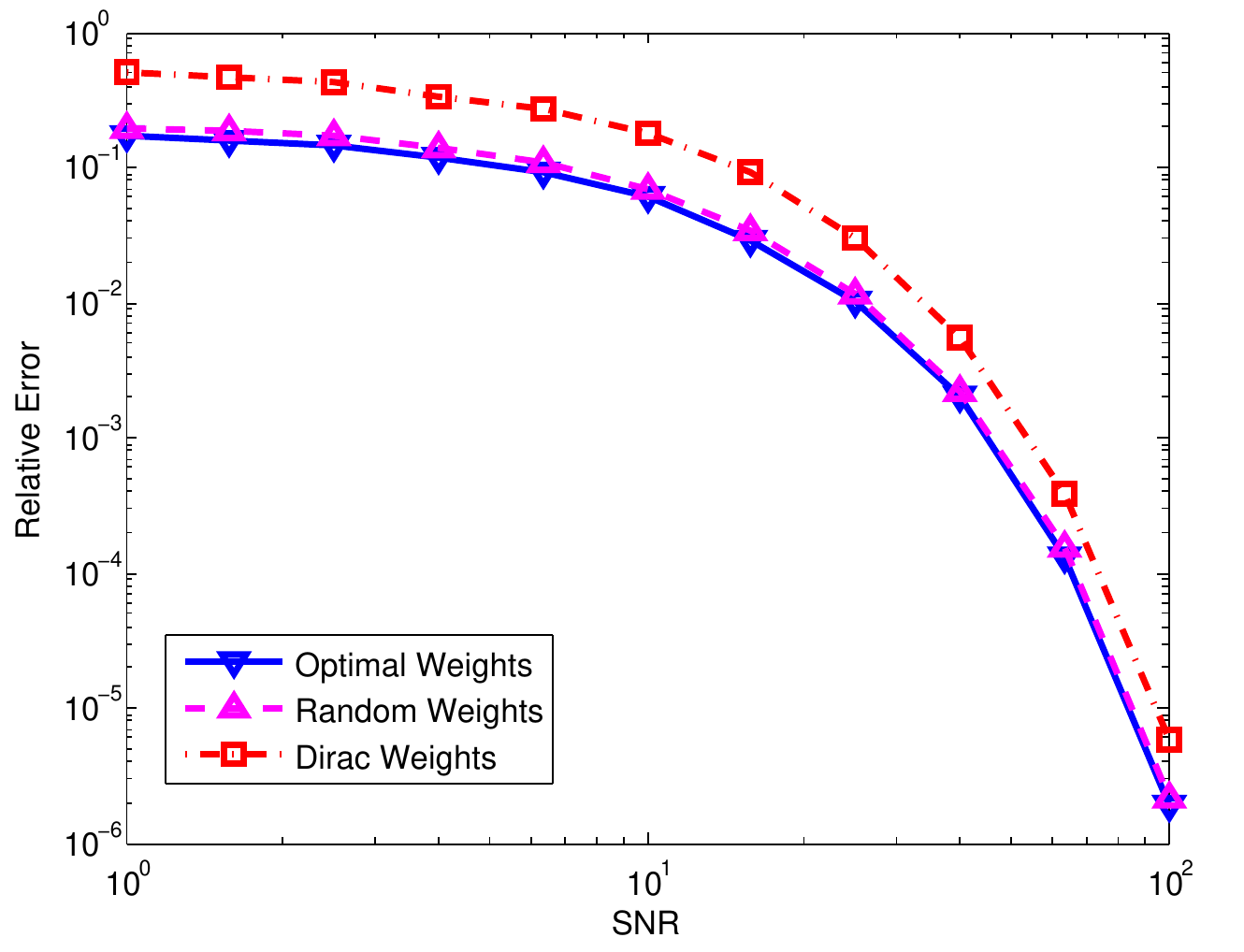}
\caption{Relative errors of ILMR under difference SNRs with various choices of local weights.
The noise associated with each vertex is \emph{i.i.d.} Gaussian.}
\label{exp2}
\end{center}
\end{figure}

In this experiment, the performance of the proposed algorithm against \emph{i.i.d.} Gaussian noise are tested for three kinds of local weights including case 1), 2), and 3). In this case the optimal local weights is equivalent to uniform weights. The graph is partitioned into $358$ centerless local sets with $N_{\rm max}$ equals $8$. The relative reconstruction errors of three tests are illustrated in Fig. \ref{exp2}. Each point is the average of $100$ trials.
The experimental result shows that for \emph{i.i.d.} Gaussian noise, reconstruction with uniform weight or random weight performs beyond that with Dirac delta weight, which is actually the traditional sampling scheme of decimation. It shows that compared with decimation, the proposed generalized sampling scheme is more robust against noise, as analyzed in section \ref{secnoise}.

\subsection{Reconstruction of Approximated Bandlimited Signals}
\begin{figure}[t]
\begin{center}
\includegraphics[width=9cm]{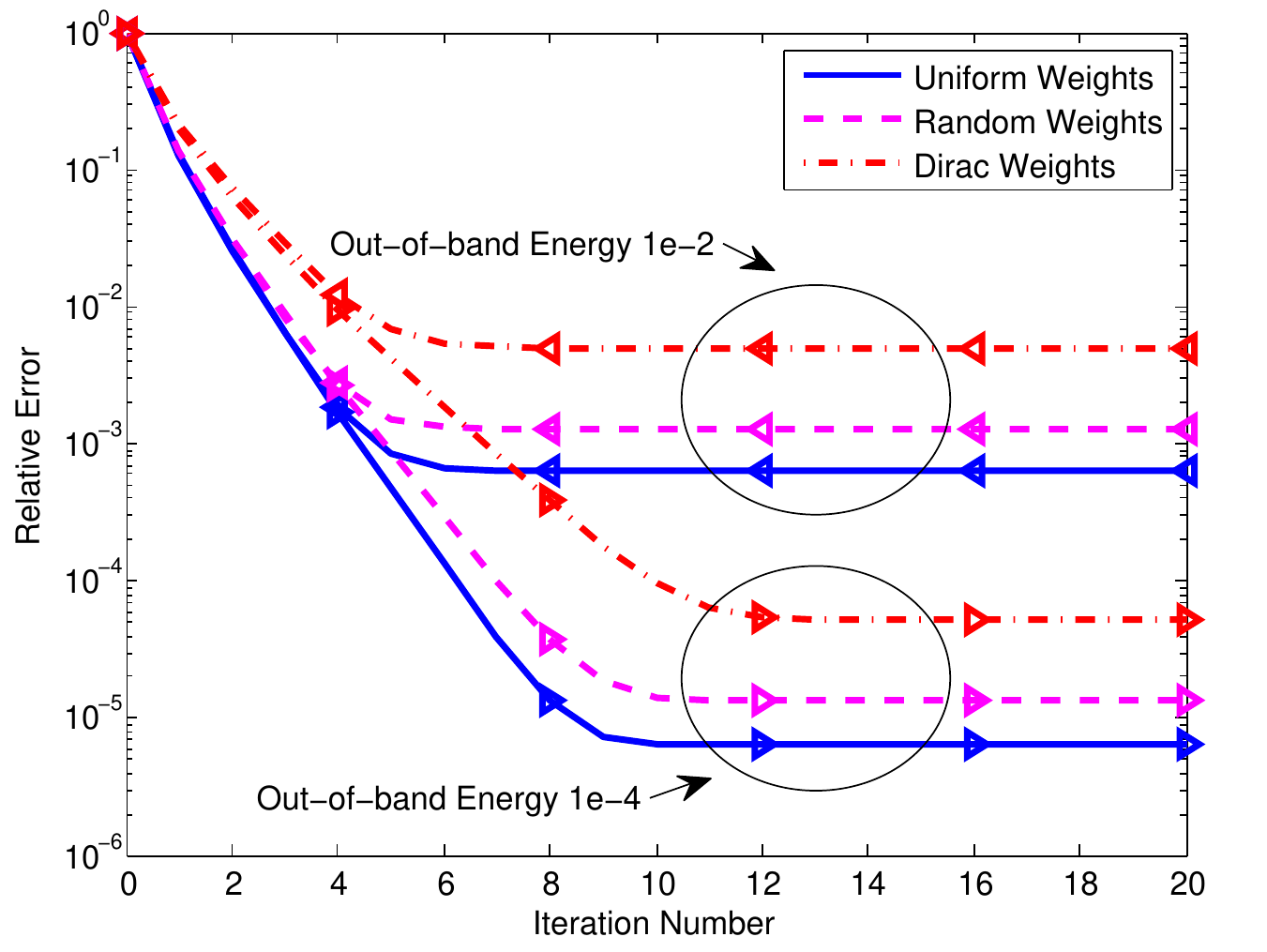}
\caption{The convergence curves for uniform weights, random weights, and Dirac delta weights if the original graph signals are approximated bandlimited.}
\label{exp4}
\end{center}
\end{figure}

In this experiment, approximated bandlimited signals are tested to be reconstructed by ILMR.
The original signal is normalized to have norm $1$ and the out-of-band energy is $10^{-2}$ or $10^{-4}$. The graph is partitioned into $358$ centerless local sets and the maximal cardinality of local sets is $8$. Three kinds of local weights are tested including case 1), 2), and 3). The convergence curves are shown in Fig. \ref{exp4}, where each curve is the average of $100$ trials. It is natural to see that the steady-state error is larger for a larger out-of-band energy. Besides, the case with uniform local weights has a smaller relative error, much better than that with Dirac weights. In other words, reconstruction from local measurements performs beyond reconstruction from decimation if the original signals are not strictly bandlimited.

\section{Conclusion}

In this paper, a sampling scheme named local measurement is proposed to obtain sampled data from graph signals, which is a generalization of graph signal decimation.
Using the local measurements, a reconstruction algorithm ILMR is proposed to perfectly reconstruct original bandlimited signals iteratively. The convergence of ILMR is proved and its performance in noise scenario is analyzed.
The optimal local weights are given to minimize the effect of noise, and a greedy algorithm for local sets partition is proposed.
Theoretical analysis and experimental results demonstrate that the local measurement sampling scheme together with reconstruction method is more robust against additive noise.

\section{Appendix}
\subsection{Proof of Lemma \ref{lemma1}}
\label{proof1}

By the definition of ${\bf G}$, and considering that $\{\mathcal{N}_i\}_{i\in\mathcal{I}}$ are disjoint, one has
\begin{align}\label{lem1-1}
\|{\mathbf f}-{\bf G}{\mathbf f}\|^2
=&\left\|P_{\omega}\left(\sum_{i\in \mathcal{I}}\left({\mathbf f}_{\mathcal{N}_i}-\langle {\mathbf f}, \bm{\varphi}_i\rangle\bm{\delta}_{\mathcal{N}_i}\right)\right)\right\|^2\nonumber\\
\le&\left\|\sum_{i\in \mathcal{I}}\left({\mathbf f}_{\mathcal{N}_i}-\langle {\mathbf f}, \bm{\varphi}_i\rangle\bm{\delta}_{\mathcal{N}_i}\right)\right\|^2\nonumber\\
=&\sum_{i\in \mathcal{I}}\left\|{\mathbf f}_{\mathcal{N}_i}-\langle {\mathbf f}, \bm{\varphi}_i\rangle\bm{\delta}_{\mathcal{N}_i}\right\|^2,
\end{align}
where
$$
f_{\mathcal{N}_i}(v)=
\begin{cases}
f(v), & v\in \mathcal{N}_i;\\
0, & v\notin \mathcal{N}_i.
\end{cases}
$$
For $i\in\mathcal{I}$, one has
\begin{align}\label{lem1-4}
\|{\mathbf f}_{\mathcal{N}_i}-\langle {\mathbf f}, \bm{\varphi}_i\rangle\bm{\delta}_{\mathcal{N}_i}\|^2\nonumber
=&\sum_{v\in \mathcal{N}_i}|f(v)-\langle {\mathbf f}, \bm{\varphi}_i\rangle|^2\nonumber\\
=&\sum_{v\in \mathcal{N}_i}\left|\sum_{p\in \mathcal{N}_i}\varphi_i(p)\left(f(v)-f(p)\right)\right|^2\nonumber\\
\le&\sum_{v\in \mathcal{N}_i}\max_{p\in \mathcal{N}_i}|f(v)-f(p)|^2
\end{align}

Denote
$$p_i(v)=\arg \max_{p\in \mathcal{N}_i}|f(v)-f(p)|^2.$$
Since $\mathcal{N}_i$ is connected, there is a shortest path within $\mathcal{N}_i$ from $v$ to $p_i(v)$, which is denoted as $v\sim v_1\sim \cdots \sim v_{k_v} \sim p_i(v)$, and the length of this path is not longer than $D_i$.
Then for $v\in \mathcal{N}_i$, one has
\begin{align}%\label{lem1-3}
\max_{p\in \mathcal{N}_i}|f(v)-f(p)|^2
=|f(v)-f(p_i(v))|^2
\le &\left(|f(v)-f(v_1)|+\cdots +|f(v_{k_v})-f(p_i(v))|\right)^2 \nonumber\\
\le &D_i\left(|f(v)-f(v_1)|^2+\cdots +|f(v_{k_v})-f(p_i(v))|^2\right).\nonumber
\end{align}
Therefore, one has
\begin{equation}\label{lem1-2}
\sum_{v\in \mathcal{N}_i}\max_{p\in \mathcal{N}_i}|f(v)-f(p)|^2
\le|\mathcal{N}_i|D_i\!\!\!\sum_{p\sim q; p,q\in \mathcal{N}_i}|f(p)-f(q)|^2,
\end{equation}
where $p\sim q$ denotes there is an edge between $p$ and $q$.
Inequality (\ref{lem1-2}) holds because each edge within $\mathcal{N}_i$ is reused for no more than $|\mathcal{N}_i|$ times. To study the right hand side of (\ref{lem1-2}), one has
\begin{align}
\sum_{p\sim q}|f(p)-f(q)|^2
=&{\mathbf f}^{\rm T}{\bf L}{\mathbf f}={\mathbf f}^{\rm T}{\bf U\Lambda U}^{\rm T}{\mathbf f}=\hat{{\mathbf f}}^{\rm T}{\bf \Lambda} \hat{{\mathbf f}}\nonumber\\
=&\sum_{\lambda_i\le\omega}\lambda_i |\hat{f}(i)|^2\le \omega\hat{{\mathbf f}}^{\rm T}\hat{{\mathbf f}}=\omega\|{\mathbf f}\|^2,\label{eqasdf}
\end{align}
where $\bf L, U$, and $\bf \Lambda$ denote the Laplacian, its eigenvectors, and its eigenvalues, respectively.
The last inequality in \eqref{eqasdf} is because the entries of spectrum $\hat{{\mathbf f}}={\bf U}^{\rm T}{\mathbf f}$ corresponding to the frequencies higher than $\omega$ are zero for ${\mathbf f}\in PW_{\omega}(\mathcal{G})$.

Consequently, utilizing \eqref{lem1-4}, \eqref{lem1-2}, and \eqref{eqasdf} in \eqref{lem1-1}, we have
\begin{align}
\|{\mathbf f}-{\bf G}{\mathbf f}\|^2\le&\sum_{i\in \mathcal{I}}\left(|\mathcal{N}_i|D_i\sum_{p\sim q; p,q\in \mathcal{N}_i}|f(p)-f(q)|^2\right)\nonumber\\
\le&C_{\rm max}^2\sum_{p\sim q}|f(p)-f(q)|^2\nonumber\\
\le&\omega C_{\rm max}^2\|{\mathbf f}\|^2\nonumber
\end{align}
and Lemma \ref{lemma1} is proved.

\subsection{Proof of Proposition \ref{pro2}}
\label{proof2}
According to Lemma \ref{lemma1}, we have $\|{\bf I-G}\|\le \gamma<1$ for $PW_{\omega}(\mathcal{G})$ when $\gamma=C_{\rm max}\sqrt{\omega}<1$.
Then ${\bf G}$ is invertible and $1-\gamma\le \|{\bf G}\|\le 1+\gamma$ for $PW_{\omega}(\mathcal{G})$.
The inverse of ${\bf G}$ is
$$
{\bf G}^{-1}=\sum_{j=0}^{\infty}({\bf I-G})^j.
$$

According to (\ref{localmeasprop2}), ${\bf f}$ can be written as
\begin{equation}\label{fe}
{\bf f}={\bf G}^{-1}{\bf Gf}
=\sum_{j=0}^{\infty}({\bf I-G})^j\sum_{i\in \mathcal{I}}\langle {\mathbf f}, \bm{\varphi}_i\rangle\mathcal{P}_{\omega}(\bm{\delta}_{\mathcal{N}_i})
=\sum_{i\in \mathcal{I}}\langle {\mathbf f}, \bm{\varphi}_i\rangle {\bf e}_i,
\end{equation}
where
$$
{\bf e}_i=\sum_{j=0}^{\infty}({\bf I-G})^j\mathcal{P}_{\omega}(\bm{\delta}_{\mathcal{N}_i}).
$$
Similarly, one has
\begin{align}
\tilde{\bf f}=\sum_{i\in \mathcal{I}}\langle \tilde{\mathbf f}, \bm{\varphi}_i\rangle {\bf e}_i.\nonumber
\end{align}

Using (\ref{iter}) and ${\bf f}^{(0)}={\bf Gf}$, we have
$$
{\bf f}^{(k)}={\bf f}+({\bf I-G})^{k}({\bf f}^{(0)}-{\bf f})={\bf f}-({\bf I-G})^{k+1}{\bf f}.
$$
Therefore
\begin{align}\label{fke}
\tilde{\bf f}^{(k)}=\tilde{\bf f}-({\bf I-G})^{k+1}\tilde{\bf f}
=\sum_{i\in \mathcal{I}}\langle \tilde{\mathbf f}, \bm{\varphi}_i\rangle {\bf e}_i
-({\bf I-G})^{k+1}\tilde{\bf f}.
\end{align}

If $\gamma=C_{\rm max}\sqrt{\omega}<1$, ${\bf e}_i$ satisfies
\begin{align}\label{norme}
\|{\bf e}_i\|\le\sum_{j=0}^{\infty}\gamma^j\|\mathcal{P}_{\omega}(\bm{\delta}_{\mathcal{N}_i})\|
\le\frac{1}{1-\gamma}\|\bm{\delta}_{\mathcal{N}_i}\|=\frac{1}{1-\gamma}\sqrt{|\mathcal{N}_i|}.
\end{align}
According to (\ref{fe}), (\ref{fke}), and (\ref{norme}),
\begin{align}
\|\tilde{\bf f}^{(k)}-{\bf f}\|
&=\left\|\sum_{i\in \mathcal{I}}\langle \tilde{\mathbf f}-{\bf f}, \bm{\varphi}_i\rangle {\bf e}_i
-({\bf I-G})^{k+1}\tilde{\bf f}\right\|\nonumber\\
&\le\sum_{i\in \mathcal{I}}|\langle {\bf n}, \bm{\varphi}_i\rangle| \left\|{\bf e}_i\right\|+\gamma^{k+1}\|\tilde{\mathbf f}\|\nonumber\\
&\le\frac{1}{1-\gamma}\sum_{i\in \mathcal{I}}\sqrt{|\mathcal{N}_i|}\cdot|n_i| +\gamma^{k+1}\left(\|{\mathbf f}\|+\|{\bf n}\|\right).\nonumber
\end{align}
Then Proposition \ref{pro2} is proved.

%\end{CJK*}

\footnotesize
\begin{thebibliography}{1}

\bibitem{shuman_emerging_2013}
D. I. Shuman, S. K. Narang, P. Frossard, A. Ortega, and P. Vandergheynst, ``The
  emerging field of signal processing on graphs: Extending high-dimensional
  data analysis to networks and other irregular domains,'' \emph{IEEE Signal
  Process. Mag.}, vol.~30, no.~3, pp. 83-98, 2013.

\bibitem{sandryhaila_big_2014}
A. Sandryhaila, and J. M. F. Moura, ``Big data analysis with signal processing on graphs: Representation and processing of massive data sets with irregular structure,'' \emph{IEEE Signal Process. Mag.}, vol. 31, no. 5, pp. 80-90, 2014.

\bibitem{zhu_graph_2012}
X. Zhu and M. Rabbat, ``Graph spectral compressed sensing for sensor networks,'' in \emph{Proc. 37th IEEE Int. Conf. Acoust., Speech, Signal Process. (ICASSP)}, 2012, pp. 2865-2868.

\bibitem{gadde_active_2014}
A. Gadde, A. Anis, and A. Ortega, ``Active semi-supervised learning using sampling theory for graph signals,'' in \emph{Proc. 20th ACM Int. Conf. Knowledge Discovery and Data Mining (KDD'14)}, 2014, pp. 492-501.

\bibitem{yang_gesture_2014}
Z. Yang, A. Ortega, and S. Narayanan, ``Gesture dynamics modeling for attitude analysis using graph based transform,'' in \emph{Proc. 21st IEEE Int. Conf. Image Process. (ICIP)}, pp. 1515-1519, 2014.

\bibitem{chen_bridge_2013}
S. Chen, F. Cerda, et al., ``Semi-supervised multiresolution classification using adaptive graph filtering with application to indirect bridge structural health monitoring,'' \emph{IEEE Trans. Signal Process.}, vol. 62, no. 11, pp. 2879-2893, 2013.

\bibitem{sandryhaila_discrete_2013}
A. Sandryhaila, and J. M. F. Moura, ``Discrete signal processing on graphs,'' \emph{IEEE Trans. Signal Process.}, vol. 61, no. 7, pp. 1644-1656, 2013.

\bibitem{chen_adaptive_2013}
S. Chen, A. Sandryhaila, J. M. F. Moura, and J. Kova\v{c}evi\'{c}, ``Adaptive graph filtering: Multiresolution classification on graphs,'' in \emph{Proc. 1st IEEE Global Conf. Signal and Inform. Process. (GlobalSIP)}, 2013, pp. 427-430.

\bibitem{Coifman_Diffusion_2006}
R. R. Coifman and M. Maggioni, ``Diffusion wavelets,'' \emph{Appl. Comput. Harmonic Anal.}, vol. 21, no. 1, pp. 53-94, 2006.

\bibitem{hammond_wavelets_2011}
D. K. Hammond, P. Vandergheynst, and R. Gribonval, ``Wavelets on graphs via spectral graph theory,''
\emph{Appl. Comput. Harmonic Anal.}, vol. 30, no. 2, pp. 129-150, 2011.

\bibitem{narang_perfect_2012}
S. K. Narang and A. Ortega, ``Perfect reconstruction two-channel wavelet filter-banks for graph structured data,'' \emph{IEEE Trans. Signal Process.}, vol. 60, no. 6, pp. 2786-2799, 2012.

\bibitem{zhu_approximating_2012}
X. Zhu and M. Rabbat, ``Approximating signals supported on graphs,'' in \emph{Proc. 37th IEEE Int. Conf. Acoust., Speech, Signal Process. (ICASSP)}, 2012, pp. 3921-3924.

\bibitem{nguyen_downsampling_2015}
H. Q. Nguyen, M. N. Do, ``Downsampling of signals on graphs via maximum spanning trees,'' \emph{IEEE Trans. Signal Process.}, vol. 63, no. 1, pp. 182-191, 2015.

\bibitem{agaskar_aspectral_2013}
A. Agaskar, and Y. M. Lu, ``A spectral graph uncertainty principle,'' \emph{IEEE Trans. Inform. Theory}, vol. 59, no. 7, pp. 4338-4356, 2013.

\bibitem{liu_coarsening_2014}
P. Liu, X. Wang and Y. Gu, ``Coarsening graph signal with spectral invariance,'' in \emph{Proc. 39th IEEE Int. Conf. Acoust., Speech, Signal Process. (ICASSP)}, 2014, pp. 1075-1079.

\bibitem{liu_graphcoarsening_2014}
P. Liu, X. Wang, and Y. Gu, ``Graph signal coarsening: Dimensionality reduction in irregular domain,'' in \emph{Proc. 2nd IEEE Global Conf. Signal and Inform. Process. (GlobalSIP)}, 2014, pp. 966-970.

\bibitem{ekambaram_multiresolution_2013}
V. N. Ekambaram, G. C. Fanti, B. Ayazifar, and K. Ramchandran, ``Multiresolution graph signal processing via circulant structures,'' in \emph{Proc. IEEE Digital Signal Process., Signal Process. Educ. Meeting (DSP/SPE)}, 2013, pp. 112-117.

\bibitem{shuman_aframework_2013}
D. I. Shuman, M. J. Faraji, and P. Vandergheynst, ``A framework for multiscale transforms on graphs,'' \emph{arXiv preprint arXiv:1308.4942}, 2013.

\bibitem{thanou_parametric_2013}
D. Thanou, D. I. Shuman, and P. Frossard, ``Parametric dictionary learning for graph signals,'' in \emph{Proc. 1st IEEE Global Conf. Signal and Inform. Process. (GlobalSIP)}, 2013, pp. 487-490.

\bibitem{dong_learning_2014}
X. Dong, D. Thanou, P. Frossard, and P. Vandergheynst, ``Learning graphs from signal observations under smoothness prior,'' \emph{arXiv preprint arXiv:1406.7842}, 2014.

\bibitem{anis_towards_2014}
A. Anis, A. Gadde, and A. Ortega, ``Towards a sampling theorem for signals on arbitrary graphs,'' in
\emph{Proc. 39th IEEE Int. Conf. Acoust., Speech, Signal Process. (ICASSP)}, 2014, pp. 3892-3896.

\bibitem{narang_localized_2013}
S. K. Narang, A. Gadde, E. Sanou, and A. Ortega, ``Localized iterative methods for interpolation in graph structured data,''
in \emph{Proc. 1st IEEE Global Conf. Signal and Inform. Process. (GlobalSIP)}, 2013, pp. 491-494.

\bibitem{wang_iterative_2014}
X. Wang, P. Liu, and Y. Gu, ``Iterative reconstruction of graph signal in low-frequency subspace,''
in \emph{Proc. 2st IEEE Global Conf. Signal and Inform. Process. (GlobalSIP)}, pp. 611-615, 2014.

\bibitem{wang_local_2014}
X. Wang, P. Liu, and Y. Gu, ``Local-set-based graph signal reconstruction,'' to appear in \emph{IEEE Trans. Signal Process.} available at \emph{arXiv:1410.3944}, 2014.

\bibitem{wang_distributed_2014}
X. Wang, M. Wang, and Y. Gu, ``A distributed tracking algorithm for reconstruction of graph signals,''
to appear in \emph{IEEE J. Selected Topics Signal Process.} available at \emph{arXiv:1502.0297}, 2015.

\bibitem{chen_distributed_2015}
S. Chen, A. Sandryhaila, J. M. F. Moura, and J. Kova\v{c}evi\'{c}, ``Distributed algorithm for graph signals,'' in Proc. \emph{IEEE Int. Conf. Acoust., Speech Signal Process.}, Brisbane, Apr. 2015.

\bibitem{narang_signal_2013}
S. K. Narang, A. Gadde, and A. Ortega, ``Signal processing techniques for interpolation in graph structured data,''
in \emph{Proc. 38th IEEE Int. Conf. Acoust., Speech, Signal Process. (ICASSP)}, 2013, pp. 5445-5449.

\bibitem{pesenson_sampling_2008}
I. Pesenson, ``Sampling in Paley-Wiener spaces on combinatorial graphs,''
\emph{Trans. Amer. Math. Soc.}, vol. 360, no. 10, pp. 5603-5627, 2008.

\bibitem{fuhr_poincare_2013}
H. F\"{u}hr and I. Z. Pesenson, ``Poincar\'{e} and Plancherel-Polya inequalities in harmonic analysis on weighted combinatorial graphs,'' \emph{SIAM J. Discrete Math.}, vol. 27, no. 4, pp. 2007-2028, 2013.

\bibitem{marvasti_nonuniform_2001}
F. Marvasti, \emph{Nonuniform Sampling: Theory and Practice}, Springer, 2001.

\bibitem{feichtinger_theory_1994}
H. G. Feichtinger, and K. Gr\"{o}chenig, ``Theory and practice of irregular sampling,''
\emph{Wavelets: Math. and Applicat.}, pp. 305-363, 1994.

\bibitem{grochenig_adiscrete_1993}
K. Gr\"{o}chenig, ``A discrete theory of irregular sampling,''
\emph{Linear Algebra and Its Applicat.}, vol. 193, pp. 129-150, 1993.

\bibitem{grochenig_reconstruction_1992}
K. Gr\"{o}chenig, ``Reconstruction algorithms in irregular sampling,''
\emph{Math. Comput.}, vol.59, no. 199, pp. 181-194, 1992.

\bibitem{aldroubi_nonuniform_2002}
A. Aldroubi, ``Non-uniform weighted average sampling and reconstruction in shift-invariant and wavelet spaces,''
\emph{Appl. Comput. Harmonic Anal.}, vol. 13, no. 2, pp. 151-161, 2002.

\bibitem{pesenson_poincare_2004}
I. Pesenson, ``Poincar\'{e}-type inequalities and reconstruction of Paley-Wiener functions on manifolds,''
\emph{J. Geometric Anal.}, vol. 14, no. 1, pp. 101-121, 2004.

\bibitem{feichtinger_recovery_2004}
H. Feichtinger and I. Pesenson, ``Recovery of band-limited functions on manifolds by an iterative algorithm,''
\emph{Contemporary Math.}, vol. 345, pp. 137-152, 2004.

\bibitem{chung_spectral_1997}
F. R. K. Chung, \emph{Spectral Graph Theory}, Amer. Math. Soc., 1997.

\bibitem{lipsey_practical_2001}
M. W. Lipsey and D. B. Wilson, \emph{Practical Meta-analysis}. Thousand Oaks, CA: Sage publications, 2001.

\bibitem{gleich_matlabbgl}
D. Gleich, The MatlabBGL Matlab Library [Online]. Available: http://www.cs.purdue.edu/homes/dgleich/packages/matlab\_bgl/index.html.

\end{thebibliography}
\end{document}